%% file: BC_Journal_Submit.tex
\documentclass[10pt]{IEEEtran}
\usepackage{amsfonts,amsmath,amssymb}
\usepackage{epsfig}
\usepackage{url}
\usepackage{psfrag}\usepackage{enumerate}
\newtheorem{thm}{Theorem}%[section]

\newtheorem{lem}{Lemma}

\newtheorem{defn}{Definition}%[section]
\newtheorem{example}{Example}%[subsection]
\newtheorem{remark}{Remark}

\newcommand{\Shv}{\hat{\mathbf{S}}}

\newcommand{\Lamt}{\tilde{\Lambda}}

\input{defns_ehsan.tex}

\begin{document}

\title{LQG Control Approach to Gaussian Broadcast Channels with Feedback}

\author{Ehsan Ardestanizadeh, Paolo Minero, and Massimo Franceschetti%
\thanks{Ehsan Ardestanizadeh and Massimo Franceschetti are with the Information Theory and Applications Center (ITA) of the California Institute of Telecommunications and Information Technologies (CALIT2), Department of Electrical and Computer Engineering, University of California, San Diego CA, 92093, USA. Email: ehsan@ucsd.edu, massimo@ece.ucsd.edu. Paolo Minero is with the Department of Electrical Engineering of the University of Notre Dame,  Notre Dame, IN, 46556, USA. Email: pminero@nd.edu }
\thanks{This work is partially supported by the National Science Foundation
award CNS-0546235 and CAREER award CCF-0747111.}
}

%    \and
% \authorblockN{Paolo Minero\thanks{}}
%  \authorblockA{Electrical and Computer Engineering\\
%    University of California, San Diego\\
%    La Jolla, CA, 92093-0407, USA \\
%    pminero@ucsd.edu} \and
%    \authorblockN{Massimo Franceschetti}
% \authorblockA{Electrical and Computer Engineering\\
%    University of California, San Diego\\
%    La Jolla, CA, 92093-0407, USA \\
%    massimo@ece.ucsd.edu}}
%\fi
%%%%%%%%%%%%%%%%%%%%%%%%%%%%%%%%%%%%%%%%%%%%
\maketitle
\IEEEpeerreviewmaketitle

\begin{abstract}
A code for communication over the $k$-receiver additive white Gaussian noise broadcast channel with feedback is presented and analyzed using tools from the theory of linear quadratic Gaussian optimal control. It is shown that the performance of this code depends on the noise correlation at the receivers and it is related to the solution of a discrete algebraic Riccati equation. For the case of independent noises, the sum rate achieved by the proposed code, satisfying average power constraint~$P$, is characterized as $1/2 \log (1+P\phi)$, where the coefficient $\phi \in  [1,k]$ quantifies the power gain due to the presence of feedback. When specialized to the case of two receivers, this includes a previous result by Elia and strictly improves upon the code of Ozarow and Leung. When the noises are correlated, the pre-log of the sum-capacity of the broadcast channel with feedback  can be strictly greater than one. %an additional gain in
%the pre-log factor of the sum rate can be achieved. 
It is established that for all noise covariance matrices of rank~$r$ the pre-log of the sum capacity is at most $k-r+1$ and, conversely, there exists a noise covariance matrix of rank~$r$ for which the proposed code achieves this upper bound. This generalizes a previous result by Gastpar and Wigger for the two-receiver broadcast channel. %
%The code is then used to established that if the noise covariance matrix has rank $r$, then the sum rate capacity under block transmit power constraint $P$ can s as the transmit power $P$ increases as a function of the rank of the noise covariance matrix. 

%

%In particular, it is shown that there exists a specific noise covariance matrix for which the achievable sum rate scales as $(k/2)\log (1+P)$ as the transmit power $\small P \to \infty$. This generalizes a previous result by Gastpar and Wigger for the two-receiver broadcast channel. In the case of independent noises, instead, the maximum sum-rate achieved by this code is equal to  $1/2 \log (1+P\phi)$, for some coefficient $\phi \in  [1,k]$, depending on both $k$ as well as $P$, that quantifies the rate gain due to the presence of feedback. When specialized to the case of two receivers, the presented code includes a previous result by Elia and strictly improves upon the code of Ozarow and Leung. 
\end{abstract}

\section{Introduction}
%Moreover, Bhaskaran~\cite{bhaskaran} showed that the capacity region may increase even with one-sided feedback and El Gamal~\cite{ElGamal--FB1,ElGamal--FB2} showed that feedback cannot increase the capacity of the physically degraded broadcast channel, in which the received signal by one user is the corrupted version of the signal received by another user.  

Consider the communication problem over the $k$-receiver additive white Gaussian noise (AWGN) broadcast channel (BC) with feedback depicted in Fig.~\ref{fig:AWGNBC}. Here, a sender wishes to communicate $k$~independent messages to $k$~distinct receivers that observe the sequence of transmitted signals corrupted by~$k$, possibly correlated, AWGN sequences. It is known that the presence of feedback from receivers to sender improves communication performance over broadcast channels. Specifically, Dueck~\cite{Dueck} showed, by providing a specific example, that feedback can enlarge the capacity region of additive memoryless broadcast channels where the noises at the receivers are correlated, while Ozarow and Leung~\cite{Ozarow--Leung} proved that feedback can be beneficial by providing a means of cooperation between receivers and sender  even when the noises at the receivers are independent. However, a computable characterization of the capacity region of this channel remains a long-standing open problem. 

In this paper, we construct a code, which we refer to as  {\it LQG code}, for communicating over the AWGN-BC with feedback and we characterize its performance using tools from optimal control.  We show that the minimum power needed by the LQG code for reliable transmission of messages encoded at a given set of rates depends on the correlation between the noises at the receivers and is determined by the solution of a discrete algebraic Riccati equation (DARE) which arises in the analysis of the linear quadratic Gaussian (LQG) optimal control problem. The LQG code is then used to investigate some properties of the capacity region of the AWGN-BC with feedback. 
%For a given set of rates, the minimum transmitted power $P$ required by the LQG code is determined by the solution to a discrete algebraic Riccati equation (DARE) and depends on the correlation of the noises at the receivers. 
%Therefore, we conclude that, unlike the no feedback case, the capacity region of the AWGN-BC can be larger when the noises are correlated.  

%For the special case of independent noises at receivers, Wu et al.~\cite{Sriram}) provided a similar analysis. 

%{\it implicitly} determines the achievable rates as a function of the power constraint for the LQG codes.
%
%based on an unstable linear system cont

%rolled over the same AWGN-BC (see Fig.~\ref{fig:ControlBC}), where the controller is chosen according to 

First, we consider the case of independent noises at the receivers, which is the most interesting in practice. 
%In this case, the noise covariance matrix is full rank $(r=k)$ and thus there is no chance to have degrees of freedom gain  $(i.e., \gamma \leq 1)$. Hence, we focus on the {\it power gain}. 
By solving the corresponding DARE, it is shown that the LQG code achieves sum rate equal to $1/2 \log (1+P\phi(k,P))$ under average power constraint $P$ when messages are encoded at the same rate. Here, the real coefficient $\phi(k,P) \in [1,k]$ represents the power gain compared to the no-feedback sum capacity $1/2 \log (1+P)$, and for fixed $k$ it is increasing in $P$, that is, more power allows more gain. In particular, as  $P \to \infty$, $\phi \to k$, and the sum rate tends to $1/2 \log (1+kP)$, which is the same as the sum capacity of the single-input multiple-output (SIMO) channel. Note that in the SIMO channel the receivers are co-located, whereas in the AWGN-BC the receivers are separately located but feedback links to the transmitter are available. Hence, the power gain due to feedback can be interpreted as the amount of {\it cooperation} among the receivers established through feedback, which allows the transmitter to align the signals intended for different receivers coherently and use power more efficiently.

Next, we investigate how the sum capacity, the supremum of achievable sum rates, scales as the power $P$ at the transmitter increases. If the sum capacity scales as $(\gamma/2)\log (1+P)$ as $P \to \infty$, then we refer to $\gamma$ as the pre-log\footnote{The pre-log is also known as the number of degrees of freedom as it is equal to the number of orthogonal point-to-point channels with the same sum capacity.} of the channel. %which we  as the {\it degrees of freedom} (DoF)~\cite{Tse} of the channel. %, as it denotes the number of orthogonal point-to-point (P2P) channels having the same sum rate capacity. 
We show that the pre-log of the AWGN-BC with feedback depends on the rank of the correlation matrix of the noises at the receivers. Specifically, if the rank is $r$, then the pre-log can be at most $k-r+1$. Conversely, for any $r \in \{1,\ldots, k\}$, there exists a noise covariance matrix of rank~$r$ for which the upper bound on the pre-log is tight and is achieved by the LQG  code. In particular, the pre-log is equal to $k$ for some rank-one covariance matrix. This generalizes a previous result by Gastpar and Wigger~\cite{Wigger} for the two-receiver AWGN-BC to the case of $k$ receivers.

%The Kramer code for the $k$-sender multiple access channel with feedback, described in Section~\ref{MACachievable}, can be obtained by solving an LQG control problem~\cite{EhsanControl}. 

%As $P \to \infty$, we have $R(P) \to 1/2\log (1+kP)$, which is the same as the capacity of a single input multiple output (SIMO) channel. Hence, in the high signal-to-noise (SNR) regime, $k$ separate receivers with feedback can receive the same sum rate as fully cooperating $k$ co-located receivers.            

%Reciprocity of MAC and BC

%Note that the power gain becomes more important in the low signal-to-noise ratio (SNR) regime. 

%We show that, when specialized to the case of two receivers and independent noises, this code has the same performance as the code by Elia~\cite{Elia2004}, which gives the best known inner bound to the capacity region of the two-receiver AWGN-BC with feedback. 

%performance improvement can be interpreted as power gain which is characterized by the solution of a discrete algebraic Riccati equation (DARE). 

%Note that the result of~\cite{Wigger} is based on the Ozarow--Leung (OL) code~\cite{Ozarow--Leung}, which is suboptimal. Hence, an optimal code is not required to achieve the best DoF since the power gain is not important when we consider the pre-log in the high SNR regime. 

Finally, we wish to mention some additional related works. 
%The idea of applying tools from optimal control to design communication codes over Gaussian channels with feedback 
%is not new. 
Elia~\cite{Elia2004} followed a control-theoretic approach and 
%based on the technique of Youla parameterization 
presented a linear code for the two-receiver AWGN-BC with independent noises, which outperforms the Ozarow--Leung (OL) code~\cite{Ozarow--Leung}. Our code, when specialized to the case of two receivers and independent noises, provides the same performance as Elia's code~\cite{Elia2004}. Wu et al.~\cite{Sriram} applied the LQG theory to study Gaussian networks with feedback, where the noises at the receivers are independent, but did not provide explicit solutions. Along the same lines, it has been shown in~\cite{EhsanControl} that the linear code proposed by Kramer~\cite{KramerFeedback} for the $k$-sender multiple access channel with feedback can be obtained by solving an LQG control problem.

%A paragraph to be added on MAC and reciprocity.  

%This code for the AWGN-MAC achieves the linear sum capacity for equal power constraints~\cite{Ehsan} and the sum capacity for the high SNR regime~\cite{KramerFeedback}. 

%For the AWGN point-to-point (P2P) channel, Elia~\cite{Elia2004} showed that the Schalkwijk--Kailath (SK) code~\cite{SchKai}, in which the encoder refines the uncertainty seen at the receiver by recursively sending the estimation error, can be obtained as the solution of an LQG optimal control problem.  

%proposed to interpret the encoder as an unstable linear dynamical system, the noise in the communication system as the plant's noise, and the channel input as the control sequence. In this way, 

%For the AWGN-BC with independent noises at the receivers, Elia~\cite{Elia2004} followed a control theoretic approach and presented a linear code for the special case of two receivers, which outperforms the OL code. Subsequently, The results in~\cite{Elia2004} and~\cite{Sriram} for the AWGN-BC with feedback are same as ours when specialized to the independent noises at the receivers. Along the same lines, the linear code proposed by Kramer~\cite{KramerFeedback} for the $k$-sender multiple access channel with feedback can be obtained by solving an LQG control problem~\cite{EhsanControl}. This code for the AWGN-MAC is known to achieve the linear sum capacity for equal power constraints~\cite{Ehsan} and the sum capacity for the high SNR regime~\cite{KramerFeedback}. 
 
% For P2P Gaussian channels, the LQG code is known~\cite{Elia2004} to achieve the feedback capacity and 
It is worth noting that the LQG code is derived from an optimal control for a linear system and hence is optimal among the subclass of linear codes. For the AWGN-BC with feedback, we show that the LQG code provides better performance compared to the OL code for $k=2$ and hence outperforms the Kramer code~\cite{KramerFeedback} for $k \geq 3$, which is an extension of the OL code. However, it remains to be  proven whether the LQG code achieves the feedback sum capacity. 
%
%Ozarow--Leung~\cite{Ozarow--Leung} linear code, an extension of the Schalkwijk--Kailath code to the two-receiver AWGN-BC with feedback, in which the transmitter recursively sends a linear combination of the estimation errors seen at the receivers. 
%%Hence, comparing to the optimal linear controller in our code, the linear coefficient (controller) is not picked appropriately in the Ozarow-Leung code. 
%Similarly, for $k \geq 3$ 

%In recent times, it has become apparent that
%Gaussian feedback communication channels represent a perfect venue to apply tools from optimal control to design communication schemes. The present work is placed in the framework of such realization.

%These results provide a step towards understanding the benefits of feedback in multi-user channels and how control tools can be used to design codes for feedback communication.

The rest of the paper is organized as follows. Section~\ref{SectionBCdef} presents the problem definition. Section~\ref{SectionBCP2P} discusses the point-to-point communication problem over the AWGN channel with feedback from a control theoretic perspective. This viewpoint is then generalized in Section~\ref{SectionLQGcode}, which presents the LQG code for communicating over the $k$-user AWGN-BC with feedback. The following two sections are devoted to studying the performance of the LQG code: in Section~\ref{SectionBCsymmetric} we provide the analysis for the case of independent noises at the receivers, while in Section~\ref{SectionBCDoF} we characterize the pre-log gain when the noises are correlated. Finally, Section~\ref{SectionBCcon} concludes the paper.

\section{Definitions}\label{SectionBCdef}

Consider the communication problem where a sender wishes to communicate $k$ independent messages $M_1,\ldots, M_k$ to $k$ distinct receivers by $n$ transmissions over the AWGN-BC channel with feedback depicted in Fig.~\ref{fig:AWGNBC}. At each time $i \in \{1, \ldots, n\}$, the channel outputs are given by
%\begin{align}
%Y_{ji}= X_{i} + Z_{ji}, \quad j=1,\ldots, k, \label{BCYdef}
%\end{align}
%where $X_{i}$ is the transmitted symbol by the sender, $Y_{ji}$ denotes the received symbol by receiver $j$, and  $Z_{ji}\sim \N(0,1)$ is the unit variance Gaussian noise at the $j$-th receiver, which is assumed to be independent of the transmitted messages. The noise vector $\Zv_i  = (Z_{1i}, \ldots, Z_{ki})^T \sim \N(0,K_z)$ is drawn independently identically distributed (i.i.d.) from a Gaussian distribution with zero mean and covariance matrix $K_z$. We assume that the output symbols are causally and noiselessly fed back to the sender such that the transmitted symbol $X_{i}$ at time~$i$ can depend on the messages $M_1,\ldots,M_k$, and the previously received channel output sequences $\Yv^{i-1}:=(\Yv_{1},\ldots,\Yv_{i-1})$, where $\Yv_i=(Y_{1i},\ldots,Y_{ki})^T$ denotes the collection of the $k$ channel outputs at time~$i$. 

\begin{align}
\Yv_{i}= \1 X_{i} + \Zv_{i} \label{BCYdef}
\end{align}
where $X_{i}$ is the transmitted symbol by the sender and $\1_{k \times 1}=(1,\ldots,1)^T$ is the column vector of ones of length $k$. The vector $\Yv_i :=(Y_{1i},\ldots,Y_{ki})^T$ contains the $k$ channel outputs at time~$i$, that is, $Y_{ji}$ is the channel output for the receiver $j$ at time~$i$. Similarly $Z_{ji}$ denotes the noise for the receiver $j$ at time $i$. The noise vector 
\[\Zv_i := (Z_{1i}, \ldots, Z_{ki})^T \sim \N(0,K_z) \ \tx{i.i.d.} \] 
is assumed to be independent of the transmitted messages, and independently and identically distributed (i.i.d.) from a Gaussian distribution with zero mean and covariance matrix~$K_z$. Without loss of generality, we assume that $Z_{ji} \sim \N(0,1)$ has unit variance, so the diagonal elements of $K_z$ are equal to one. 

We assume that the output symbols are causally and noiselessly fed back to the sender so that the transmitted symbol $X_{i}$ at time~$i$ can depend on the message vector $\Mv:=(M_1,\ldots,M_k)^T$, and the previous channel output vectors $\Yv^{i-1}:=(\Yv_{1},\ldots,\Yv_{i-1})$.
%, where $\Yv_i=(Y_{1i},\ldots,Y_{ki})^T$ denotes the collection of the $k$ channel outputs at time~$i$. 

\begin{figure*}[htbp]
  \psfrag{e1}[b]{\hspace{1.5em}\small  $ \Mh_1$}
\psfrag{e2}[b]{\hspace{1.5em}\small  $\Mh_j$}
\psfrag{e3}[b]{\hspace{1em}\small  $\Mh_k$}
\psfrag{m}[b]{\hspace{-6 em} \small  $\Mv=(M_1,\ldots,M_k)$}

\psfrag{en}{\hspace{-.4em} { \small Encoder}}

\psfrag{de1}{\hspace{-.2em} { \small Decoder $1$ }}
\psfrag{de2}{\hspace{-.2em} { \small Decoder $j$ }}
\psfrag{de3}{\hspace{-.2em} { \small Decoder $k$}}

\psfrag{x}{\hspace{-.7em} \small $X_i$} %\footnotesize
\psfrag{y1}{\hspace{-.4em}\small $Y_{1i}$}
\psfrag{y2}{\hspace{-.4em}\small $Y_{ji}$}
\psfrag{y3}{\hspace{-.4em}\small $Y_{ki}$}
\psfrag{z1}{\hspace{-.2em}\small  $Z_{1i} $}
\psfrag{z2}{\hspace{-.2em}\small $Z_{ji}$}
\psfrag{z3}{\hspace{-.2em}\small  $Z_{ki}$}
\psfrag{f}{\hspace{-5em} { \small $\Yv^{i-1}=(Y^{i-1}_1, \ldots, Y^{i-1}_k)$}}
%\psfrag{F}{\hspace{-1.5em} $Y^{i-1}$}        
\psfrag{d}{\hspace{.5em} {\LARGE$\vdots$}}           
          
                      \centering
           \scalebox{1}{
  
           \includegraphics[width=4 in]{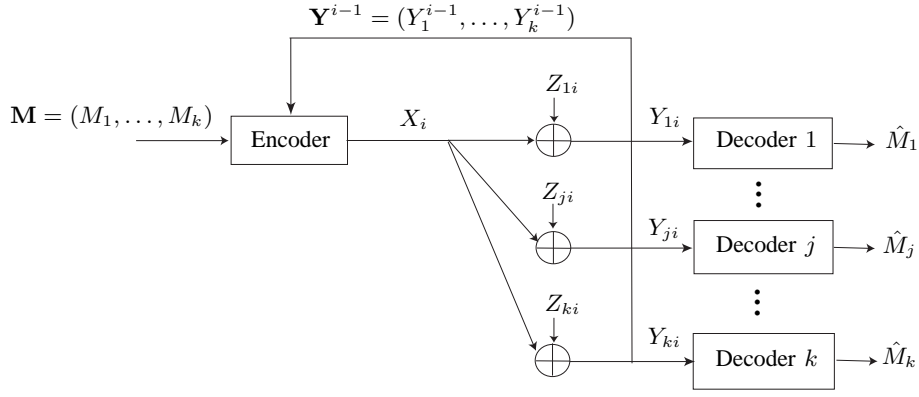}  }
          \vspace{0ex}
          \caption{The $k$-receiver AWGN broadcast channel with feedback.}
         \label{fig:AWGNBC}
\end{figure*}

\begin{defn}
A $(2^{nR_1},\ldots,2^{nR_k}, n)$ code for the AWGN-BC with feedback consists of
\begin{enumerate}
\item $k$ discrete messages $M_j \sim \U \{1,\ldots,2^{nR_j}\}$, %$j=1,\ldots,k$,
\item an encoder that assigns a symbol $X_{i}(\Mv,\Yv^{i-1})$ to the message vector $\Mv$ and the previous channel output vectors $\Yv^{i-1}$ for each $i \in \{1,\ldots,n\}$, and
\item $k$ decoders, where decoder $j$ assigns an estimate $\Mh_j$ to each sequence $Y_j^n:=(Y_{j1},\ldots,Y_{jn})$.
\end{enumerate}
\end{defn}
Let $M_1,\ldots,M_k$ be independent and the probability of error be defined as
\begin{align*}
\pen = \P(M_j \neq \Mh_j \ \tx{for some} \ j ).
\end{align*}
Then, we say that  $(R_1, \ldots, R_k)$ is an achievable rate vector under (asymptotic block) power constraint $P$  if there exists a sequence of $(2^{nR_1},\ldots,2^{nR_k},n)$ codes such that 
\begin{align*}
\lim_{n\to \infty} P^{(n)}_{e} = 0
\end{align*}
and 
\begin{align}\label{avP}
\limsup_{n \to \infty} \frac{1}{n}\sum_{i=1}^n \E(X^2_{i}) \leq P.
\end{align}
We refer to $R =\sum_{j=1}^k R_j$ as the sum rate of an achievable rate vector. 
\begin{defn}
The closure of the set of achievable rate vectors $(R_1, \ldots, R_k)$ under power constraints $P$ is called the capacity region $\Cr(P,K_z)$. The sum capacity $C(P,K_z)$ is defined as
\[C(P,K_z) := \max \left\{ \sum_{j=1}^k R_j : (R_1, . . . ,R_k) \in \Cr(P,K_z) \right\} \]
and the pre-log $\gamma(K_z)$ is defined as
\begin{align*}
%\label{dof}
\gamma(K_z) = \limsup_{P \to \infty} \frac{ C(P,K_z) }{\frac{1}{2} \log (1 + P )}.
\end{align*}
\end{defn}

%Recall that the definition of achievable MSE exponent tuple assumes that the messages $M_1,\ldots, M_k$ are drawn independently and uniformly  from the unit interval,  which does not depend on the block length~$n$. 
%The following lemma provides the connection between the achievable rates and MSE exponents.
%The following lemma establish that a sequence of $n$ codes that achieves a certain MSE exponent vector can be used to construct a sequence of $n$ codes achieving a set of rate vectors over the same channel. 
%a sufficient condition for a rate tuple to be sufficient under some block power constraint.
\begin{defn}
An $n$-code for the AWGN-BC with feedback consists of
\begin{enumerate}
\item $k$ continuous messages $\Theta_j \sim \U (0,1)$
\item an encoder that assigns a symbol $X_{i}(\Thetav,\Yv^{i-1})$ to the message vector $\Thetav=(\Theta_1,\ldots,\Theta_k)^T$ and the previous channel output vectors $\Yv^{i-1}$ for each $i \in \{1,\ldots,n\}$, and
\item $k$ decoders, where decoder $j$ assigns an estimate $\Thetah_j$ to each sequence $Y_j^n=(Y_{j1},\ldots,Y_{jn})$.
\end{enumerate}
\end{defn}
Let $\Theta_1,\ldots,\Theta_k$ be independent and the mean square errors $\Dn_1,\ldots, \Dn_k$ at time $n$ be defined as
\[\Dn_j =\E (\Theta_j-\Thetah_j)^2, \quad j=1,\ldots, k.\]
Then, we say that $(E_1,\ldots,E_k)$ is an achievable mean square error (MSE) exponent vector under (asymptotic block) power constraint $P$  if there exists a sequence of $n$-codes such that
\[E_j = \lim_{n \to \infty} -\frac{1}{2n}\log \Dn_j, \quad j=1,\ldots, k \]
and~\eqref{avP} holds.
%\begin{align}\label{asymP}
%\limsup_{n \to \infty}\frac{1}{n} \sum_{i=1}^n \E(X^2_{n}) \leq P.
%\end{align}

The definitions of achievable MSE exponent and rate vectors are closely related, as established by the following lemma. 
\begin{lem}\label{lemMSE}
Let $(E_1,\ldots,E_k)$ be an achievable MSE exponent vector under power constraint $P$, and $(R_1,\ldots,R_k)$ be such that  
\begin{align*}%\label{twoapp}
R_j< E_j, \ j=1,\ldots, k.
\end{align*}
Then, the rate vector $(R_1,\ldots,R_k)$ is achievable under power constraint $P$. %is asymptotically satisfied. 
\end{lem}
\IEEEproof See Appendix~\ref{applemMSE}.

%According to Lemma~\ref{lemMSE}, showing that MSE exponent $E_j > R_j$ is achievable for the uniform message $M_j \in (0,1)$ is sufficient to show the achievability of rate $R_j$, defined based on discrete message sets $\Mc_{j,n}:=\{1,2,\ldots,2^{nR_j}\}$.

%%%%%%%%%%%%%%%%%%%%%%%%%%%%iffalse%%%%%%%%%%%%%%%

\section{LQG Approach: Point-to-Point Channels}\label{SectionBCP2P}
%In this section, we describe how feedback codes can be designed for communication over a point-to-point AWGN channel based on an unstable system stabilized over the same channel. 
%In Section~\ref{PrelimControlComm}, it was shown that the combination of the system and the control for the problem depicted in Figure~\ref{fig:controlAWGN} determines an encoder for the feedback communication over an AWGN channel shown in Figure~\ref{fig:P2PAWGN}.

Before presenting the LQG code for the AWGN-BC, we first revisit the communication problem over the point-to-point AWGN channel with feedback in Fig.~\ref{fig:P2PAWGN}, and demonstrate how the theory of LQG control can be used to design codes for communication over such channel.  It is well known~\cite{Elia2004} that the capacity-achieving code by Schalkwijk and Kailath (SK)~\cite{SchKai, Sch} can be derived as the solution of an optimal control problem. However, here we provide a derivation of this result that naturally generalizes to the case of multiple receivers.   

%Consider the point-to-point AWGN channel with feedback in Fig.~\ref{fig:P2PAWGN}.  We show that the well-known capacity-achieving code by Schalkwijk and Kailath (SK)~\cite{SchKai, Sch} can be derived from the theory of the LQG control problem. This observation was first made by Elia~\cite{Elia2004}.  

Let $S_1 \in \Real$ be the initial state of an unstable linear system with open-loop dynamics  
\[S_i=aS_{i-1}, \quad i=2,3, \ldots\] 
for some coefficient $a >1$, that is stabilized by a controller having full state information and where the control signal is corrupted by AWGN (see Fig.~\ref{fig:BCcontrolAWGN}). %stabilized over an AWGN channel (see Figure~\ref{fig:BCcontrolAWGN}).  
The closed-loop dynamic of this system is given by 
\begin{align}\label{eq:BCP2Pclosedloop}
S_i=aS_{i-1}+Y_{i-1}, \quad i=2,3, \ldots
\end{align}
where 
\begin{align}
\label{eq:AWGN_}
Y_i=X_i+Z_i, \quad Z_i \sim \N(0,1),
\end{align} 
and
\begin{align}\label{eq:BCP2Pcontrol}
X_i=\pi_i(S_i), \quad i=1,2,\ldots
\end{align}
The mappings $\{\pi_i\}_{i=1}^\infty$ are referred to as the {\it control}, and the linear dynamical system in~\eqref{eq:BCP2Pclosedloop}, which is characterized by the unstable mode $a$, is simply referred to as the {\it system}. We say that a control is stabilizing if $\limsup_{i \to \infty} \E(S^2_i) < \infty.$

Given a control ~\eqref{eq:BCP2Pcontrol} for the system~\eqref{eq:BCP2Pclosedloop}, we construct a sequence of $n$-codes for the point-to-point AWGN channel~\eqref{eq:AWGN_} with feedback (cf. Definition 3 in the special case of $k=1$) as follows. 
\begin{enumerate}
%\item Encoder: Having access to the previous channel outputs, the encoder recursively forms $S_i$ as in~\eqref{eq:BCP2Pclosedloop} and according to the control $\{\pi\}$ transmits $X_i=\pi(S_i)$ for each $i \in \{1,\ldots,n\}$,
%\item Decoder: the decoder recursively forms 
%\[S'_i=aS'_{i-1}+Y_{i-1},  \quad S'_0=0\] 
%and chooses $\Theta_n=-a^{-n}S'_n$.
\item Encoder: Given a continuous message $\Theta \sim \U (0,1)$, the encoder recursively forms 
\begin{align}\label{P2Penc}
S_1&=\Theta \nonumber \\
S_i&=aS_{i-1}+Y_{i-1}, \quad i=2,3,\ldots,n
\end{align}
and transmits $X_i(\Theta,Y^{i-1})=\pi_i(S_i)$ for each $i \in \{1,\ldots,n\}$.
\item Decoder: the decoder sets $\Thetah_i=-a^{-i}\Sh_{i+1}$ as an estimate of the message $\Theta$, where $\Sh_i$ is recursively formed as
\begin{align}\label{P2Pdec}
\Sh_1&=0 \nonumber \\
\Sh_i&=a\Sh_{i-1}+Y_{i-1},  \quad i=2,3,\ldots,n.
\end{align} 
\end{enumerate}

\begin{figure*}[htbp]
\vspace{.2in}
\centering

     \psfrag{f}[b]{ \hspace{7 em}\small $S_{i}=g_i(S_{i-1},Y_{i-1})$}
\psfrag{s}[b]{ \hspace{6.5em} System}
%\psfrag{i}[b]{\small  $S_0$}
\psfrag{y}[b]{\hspace{0em} \small $Y_i$}
\psfrag{u}[b]{\hspace{0em} \small $X_i$}
\psfrag{z}[b]{\hspace{.3em} \small  $Z_i$}
\psfrag{m}[b]{\hspace{.3em} \small  $M$}
\psfrag{t}[b]{\hspace{.3em} \small  $\Mh$}

\psfrag{c}[b]{\hspace{4em} \small Encoder}
\psfrag{f}[b]{\hspace{3em} \small Decoder}

   \includegraphics[width=4 in]{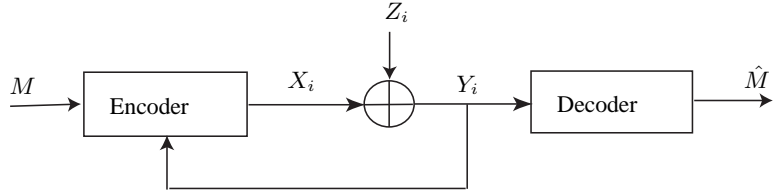} 
\vspace{.2in}   
  \caption{The AWGN channel with feedback.} \label{fig:P2PAWGN}
  \end{figure*}

\begin{figure*}[tbp]
 \psfrag{f}[b]{ \hspace{3.5em}\small $S_{i}=a S_{i-1}+Y_{i-1}$}
\psfrag{s}[b]{ \hspace{3.5em} System}
\psfrag{i}[b]{\hspace{0 em}\small  $S_1$}
\psfrag{y}[b]{\hspace{3em} \small $Y_i$}
\psfrag{u}[b]{\hspace{-2em} \small $X_i$}
\psfrag{z}[b]{\hspace{.3em} \small  $Z_i\sim \N(0,1)$}
\psfrag{q}{\hspace{-.5em} { \small Controller}}
\psfrag{e}[b]{ \hspace{3.5em} \small  $X_i=\pi_i(S_{i})$}
                      \centering
                      \hspace{-1in}
                                       
           \scalebox{1}{
  
  \includegraphics[width=3 in]{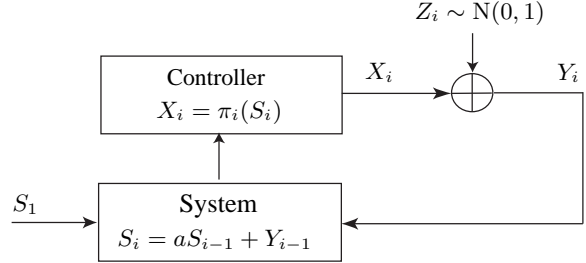}   }      
          \vspace{2ex}
          \caption{Stabilization over the AWGN channel.}
         \label{fig:BCcontrolAWGN}
\end{figure*}

\begin{lem}\label{stP2P}
If the control $\{\pi_i\}$ stabilizes the system with open-loop mode $a > 1$, then the $n$-code described by~\eqref{P2Penc} and~\eqref{P2Pdec} achieves MSE exponent $E=\log a$ under power constraint 
\begin{equation}
\label{eq:powc}
P=\limsup_{n\to \infty} \frac{1}{n} \sum_{i=1}^n\E(\pi^2_i(S_i)).
\end{equation}
\end{lem}
\begin{proof}
Combining~\eqref{P2Penc} and \eqref{P2Pdec} we have
\begin{align}
S_{i+1}&=a^i \Theta+\Sh_{i+1}\nonumber \\
&=a^i(\Theta-\Thetah_i) \label{P2Perr}
\end{align}
where the last equality follows from the fact that $\Thetah_i=-a^{-i}\Sh_{i+1}$. From~\eqref{P2Perr}, the MSE of the estimate $\Thetah_n$ is $\Dn=\E(\Theta-\Thetah_n)^2=a^{-2n} \E(S_{n+1}^2)$. Since the control is stabilizing we know $\limsup_{n \to \infty} \E(S_n^2) < \infty$ and hence
\[E=\lim_{n \to \infty} -\frac{1}{2n} \log \Dn = \log a.\]
Noting that $X_i=\pi_i(S_i)$, we conclude that MSE exponent $\log a$ is achievable under asymptotic power equal to $\limsup_{n\to \infty} \frac{1}{n} \sum_{i=1}^n\E(\pi^2_i(S_i))$.
\end{proof}

Lemma~\ref{stP2P} shows that for a fixed $a > 1$ any stabilizing control yields a code achieving MSE exponent $E=\log a$. In general, though, different codes require a different asymptotic power constraint~\eqref{eq:powc}. Thus, we are interested in finding the stabilizing control whose associate asymptotic power is minimal. 
%Next, we find the minimum power required to stabilize the system $a$, or equivalently, to achieve the MSE exponent $E=\log |a|$. 
It is known from the theory of LQG control~\cite{Chen} that the linear stationary control
\begin{align}\label{P2Pstlin}
X_i= - cS_i , \quad i=1,2,\ldots
\end{align}
where $c=(a^2-1)/a$ attains the minimum asymptotic power~\eqref{eq:powc}, which is given by
\[P^*= a^2-1.\] 
Hence, from Lemma~\ref{stP2P}, the linear code corresponding to this optimal LQG control, which we refer to as the LQG code, achieves the MSE exponent  
\[\log a= \half \log (1+P^*)\] 
under power constraint $P^*$. By Lemma~\ref{lemMSE}, when specialized to $k=1$, we conclude that the LQG code achieves any rate $R< \log a=1/2\log(1+P^*)$ under power constraint $P^*$, and hence is capacity achieving.  

A natural question to ask is what is the connection between the LQG code and the SK code, where the sender recursively transmits the estimation error at the receiver? To answer this question, note that combining~\eqref{P2Penc} and~\eqref{P2Pstlin} we can write a recursion for the channel input $X_i$ as
\begin{align}\label{LQGrecursion}
X_{i+1}&= a\left(X_{i} - \frac{c}{a} \ Y_{i}\right)
\end{align}
%On the other hand, the SK code can be represented as~\cite{YH--lecture}
%\begin{align}
%X_{i+1}&= a(X_{i} -X'_i(Y_i)), \quad i=2,\ldots    \\
%X'_i(Y_i)&= \frac{\E(X_iY_i)}{\E(Y_i^2)} \ Y_i 
%\end{align}
%with $X_1=\Theta$ and $X_1'=0$. 
with $X_1=-c\Theta$. The recursion converges, irrespectively of the initial value, since $a-c = 1/a < 1$.  On the other hand, the channel input update in the SK can be represented by the following recursion~\cite{YH--lecture},
\begin{align}\label{SKrecursion}
X_{i+1}&= a\left(X_{i} - \frac{\E(X_iY_i)}{\E(Y_i^2)} \ Y_{i} \right)
\end{align}
with $X_1=\Theta$. Comparing~\eqref{LQGrecursion} and~\eqref{SKrecursion}, we can see that the LQG code is asymptotically equivalent to the SK code if   
\begin{align}\label{Equivalence}
\frac{\E(X_iY_i)}{\E(Y_i^2)} \to  \frac{c}{a}  \ \tx{as} \ i \to \infty
\end{align}
such that, in steady state, $(c/a)Y_i$ tends to the minimum mean squared error estimate of $X_i$ given $Y_i$. 

%To show~\eqref{Equivalence}, 
By plugging~\eqref{P2Pstlin}    
into~\eqref{P2Penc}, 
we have the closed-loop dynamics for $S_i$ as 
\[S_i=a^{-1}S_{i-1}+Z_{i-1}.\]
As $i \to \infty$ the second moment of the state converges to  $a^2/(a^2-1)$ and since $X_i=-cS_i$,  
\begin{align*}
\E(X_i^2)&\to a^2-1 \\
\E(Y_i^2) &\to a^2.
\end{align*}
Therefore, as $i \to \infty$,
\begin{align*}
\frac{\E(X_iY_i)}{\E(Y_i^2)}&=\frac{\E(X_i^2)}{\E(Y_i^2)}\to  \frac{a^2-1}{a^2} = \frac{c}{a}
\end{align*}
where the last equality follows from the fact that the optimal control is given by $c=(a^2-1)/a$. Hence, the LQG code and the SK code are asymptotically equivalent.

\section{LQG Code: AWGN Broadcast Channel with Feedback}\label{SectionLQGcode}

\begin{figure*}[tbp]
 \psfrag{e1}[b]{\hspace{1.5em}\small  $ \Mh_1$}
\psfrag{e2}[b]{\hspace{1.5em}\small  $\Mh_j$}
\psfrag{e3}[b]{\hspace{1em}\small  $\Mh_k$}
\psfrag{m}[b]{\hspace{-5 em} \small   $\Sv_1=(S_{11},\ldots,S_{k1})$}
\psfrag{s}[b]{ \hspace{2.8 em} System}
\psfrag{en}{\hspace{-3em} { \small$\Sv_{i}=A\Sv_{i-1}+\Yv_{i-1}$}}
\psfrag{q}{\hspace{-1em} { \small Controller}}
\psfrag{e}[b]{ \hspace{2.5em} \small  $X_i=\pi_i(\Sv_{i})$}

\psfrag{de1}{\hspace{-.2em} { \small Decoder $1$ }}
\psfrag{de2}{\hspace{-.2em} { \small Decoder $j$ }}
\psfrag{de3}{\hspace{-.2em} { \small Decoder $k$}}
\psfrag{x}{\hspace{-.7em} \small $X_i$} %\footnotesize
\psfrag{y1}{\hspace{-.4em}\small $Y_{1i}$}
\psfrag{y2}{\hspace{-.4em}\small $Y_{ji}$}
\psfrag{y3}{\hspace{-.4em}\small $Y_{ki}$}
\psfrag{z1}{\hspace{-.2em}\small  $Z_{1i} $}
\psfrag{z2}{\hspace{-.2em}\small $Z_{ji}$}
\psfrag{z3}{\hspace{-.2em}\small  $Z_{ki}$}
%\psfrag{f}{\hspace{-3em} { \small $\Yv^{i-1}=(Y^{i-1}_1, \ldots, Y^{i-1}_k)$}}
%\psfrag{F}{\hspace{-1.5em} $Y^{i-1}$}        
\psfrag{d}{\hspace{.5em} {\LARGE$\vdots$}}

                      \centering
                      \hspace{-1in}
                                       
           \scalebox{1}{
  
                \includegraphics[width=4 in]{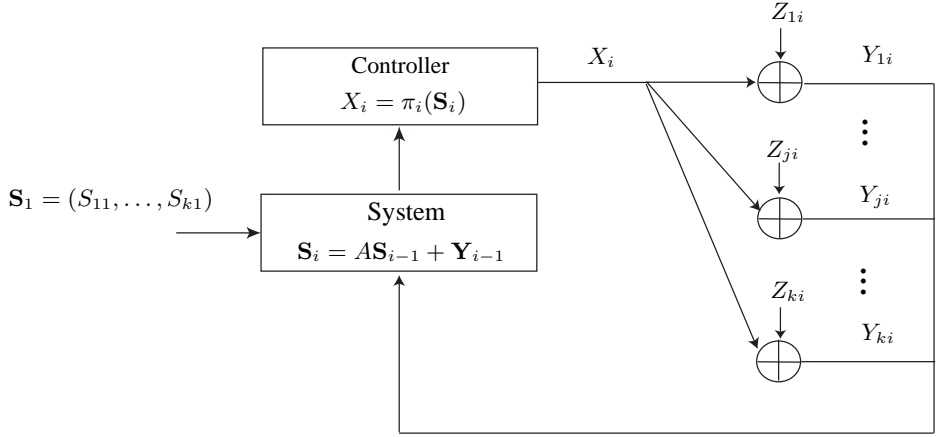}  }
          \vspace{2ex}
          \caption{Control over the AWGN broadcast channel}
         \label{fig:ControlBC}
\end{figure*}

In this section we extend the control theoretic approach to the case of a $k$-receiver AWGN-BC with feedback. We do so by considering the control problem depicted in Fig.~\ref{fig:ControlBC}, in which a $k$-dimensional unstable dynamical system is stabilized by a scalar controller having full state observation and where the scalar controller is perturbed by $k$, possibly correlated, AWGN noises, each affecting a different component of the state vector. We show that any controller stabilizing the system in mean-square sense yields a code for the $k$-receiver AWGN-BC with feedback. In particular, the LQG code is obtained from the minimum variance stabilizing control which can be computed using the LQG control theory. 
%Given a general control for an unstable linear system, we present a code for
%the AWGN-BC with feedback in a similar manner as in Section~\ref{SectionBCP2P}. Considering the stabilizing control with minimum power, which can be derived using the theory of LQG optimal control, we present and analyze the LQG code for the AWGN-BC. 
%
%juxtapose the communication problem over the $k$-receiver AWGN-BC with feedback depicted in Fig.~\ref{fig:AWGNBC} with the control problem over the same AWGN-BC channel depicted in Fig.~\ref{fig:ControlBC}. Given any stabilizing controller for the latter, we present an $n$-code for
%the AWGN-BC with feedback. Considering the optimal controller with minimum power, which can be derived using the theory of LQG optimal control, we present the LQG codes for the AWGN-BC and provide the performance analysis. % we construct linear feedback codes for the AWGN-BC based %The analysis is based on the solution of an LQG control problem. 
\subsection{Code Design Based on a Control Approach}
Assume for simplicity that channel input and outputs in~\eqref{BCYdef} are complex numbers and that the additive noise vector is drawn i.i.d. from a circular symmetric complex Gaussian distribution with covariance $K_z$. It is easy to see that if $(R_1,\ldots,R_k)$ is achievable under power constraint $P$ per each real dimension of the complex channel, then $(R_1,\ldots,R_k)$ is achievable under the same power constraint over the original real channel. In fact, one transmission over the complex channel can be reproduced by two consecutive transmission (of real and imaginary part, respectively) over the real channel.  
%
%Consider the following $n$ code for~\eqref{channel}
%\begin{enumerate}
%\item For each $i=1,\ldots,n$, let $X_{i}=\pi_{i}(\Sv_i)$ for some complex function $\pi_i$ and state variable $\Sv_i$

%
%the encoder recursively forms
%\begin{align}
% \Sv_{0} & = \Mv, \Yv_0 =0, \nn \\
%\Sv_i & = A \Sv_{i-1}+ \Yv_{i-1},  i=1, \ldots, n \label{encoder}
%\end{align}
%and transmits $X_{i}=\pi_{i}(\Sv_i)$, where $\pi_i$ is given by \eqref{control}.
%%Then we have
%%\[\sum_{j=1}^k x_{ji}=-C\sv_i=X_i\]
%%where the last equality follows from (\ref{control}).
%\item Decoder: At the end of $n$ channel uses, the $j$-th decoder forms the estimate of the transmitted message 
%$$\Mh_{jn}:=-(\b_j\omega_j)^{-n}\Sh_{jn}, \nn $$ 
%where $\Sh_{jn}$ is recursively formed by
%\begin{align}
%\Sh_{j0} & =0, Y_{j0}=0, \nn \\
%\Sh_{ji}&=A \Sh_{j(i-1)}+Y_{j(i-1)}, \quad i =1, \ldots, n\label{decoder}
%\end{align}
%\end{enumerate}
%Note that a transmission over a complex channel can be viewed as two transmissions over the real channel and the achievable rates per each complex dimension are achievable over the real channel.
Let
\begin{align}\label{BCAform}
A &=  \text{diag}(a_1,\ldots,a_k) \in \mathbb{C}^{k \times k} 
\end{align}
where $a_j \in \Complex, \ j=1\,\ldots,k$, are distinct points outside the unit circle, i.e., $|a_j|>1$, and $A'$ denote the {\em conjugate transpose} of the matrix $A$. Consider the linear dynamical system with open-loop matrix $A$ shown in Fig.~\ref{fig:ControlBC},
\begin{align}\label{BCsystem}
\Sv_1 & = \Thetav \nn \\
\Sv_{i}&=A\Sv_{i-1}+\Yv_{i-1}, \quad i = 2,3,\ldots
\end{align}
where $\Sv_{i} = (S_{1i},\ldots, S_{ki})^T \in \Complex^k$ represents the state of the system at time $i$, $\Thetav=(\Theta_1,\ldots,\Theta_k)^T$ is a vector of complex random variables such that $\Theta_1,\ldots,\Theta_k$ are drawn i.i.d. from a uniform distribution over $(0,1)\times (0,1) \subset \mathbb{C}$, and $\Yv_i \in \Complex^k$ denotes the vector of complex channel outputs, i.e.,
\begin{align}\label{channel}
\Yv_{i}&=B X_{i} + \Zv_{i}, \quad \Zv_{i} \sim \Cc\Nc(0, K_z)
\end{align}
where $X_i \in \mathbb{C}$ here represents the scalar complex control signal, $\Zv_{i}\in \Complex^k$ denote the noise vector at time~$i$, and
\begin{align}\label{Bform}
B&= [1,\ldots,1]^\prime_{1 \times k}.  %\in \mathbb{C}^{k \times 1}
\end{align}
At every discrete time $i$, the control input can depend only on the state of the system at time $i$, so
\begin{align}\label{BCcontrol}
X_{i}=\pi_{i}(\Sv_i), \quad i=1,2, \ldots
\end{align}
for some function $\pi_i : \Complex^k \to \Complex$. We refer to the sequence $\{\pi_i\}$ as the control. Since $|a_j|>1$, the eigenvalues of $A$ are outside the unit circle and the open-loop system~\eqref{BCsystem} is unstable. We say that the control $\{\pi_i\}$ stabilizes the closed-loop system if
\[\limsup_{n \to \infty} \E( \|\Sv_n\|^2) < \infty\]
where $\|\Sv_n\|^2$ denotes the $2$-norm of the vector $\Sv_n$.

%\begin{figure*}[htbp]
% 

%
% \psfrag{f}[b]{ \hspace{3.5em}\small $\Sv_{i}=A\Sv_{i-1}+BY_{i-1}$}
%\psfrag{s}[b]{ \hspace{3.5em} System}
%\psfrag{i}[b]{\hspace{-6em}\small  $\Sv_0=(S_{10},\ldots,S_{N0})$}
%\psfrag{y}[b]{\hspace{3em} \small $Y_i$}
%\psfrag{u}[b]{\hspace{-2em} \small $X_i$}
%\psfrag{z}[b]{\hspace{.3em} \small  $Z_i\sim \N(0,1)$}
%\psfrag{q}{\hspace{-.5em} { \small Controller}}
%\psfrag{ch}[b]{\hspace{7em} \small Channel}

%\psfrag{e}[b]{ \hspace{3.5em} \small  $X_i=\pi_i(\Sv_{i})$}

%          
%                      \centering
%                      \hspace{-1in}
%                                       
%           \scalebox{1}{
%  
%           \includegraphics[width=2.8in]{control_AWGN.eps} } 
%          \vspace{0ex}
%          \caption{Stabilization over a point-to-point Gaussian channel}
%         \label{fig:AWGNcontrolBC}
%\end{figure*}

Given the system~\eqref{BCsystem} and the control~\eqref{BCcontrol}, we present the following sequence of $n$-codes for the $k$-receiver AWGN-BC with feedback~\eqref{channel}.
%Given the system \eqref{BCsystem} and the controller \eqref{control}, we derive an $n$-code for the AWGN-BC as follows.
\begin{enumerate}
\item Encoder: At each time $i \in \{1, \ldots, n\}$ the encoder recursively forms $\Sv_i$ as in~\eqref{BCsystem} and transmits $X_{i}=\pi_{i}(\Sv_i)$.
%, where $\Sv_i$ is formed recursively as in~\eqref{BCsystem} and $\pi_1,\ldots,\pi_k$ is a stabilizing control for~\eqref{BCsystem} .

%
%: the encoder recursively forms
%\begin{align}
% \Sv_{0} & = \Mv, \Yv_0 =0, \nn \\
%\Sv_i & = A \Sv_{i-1}+ \Yv_{i-1},  i=1, \ldots, n \label{encoder}
%\end{align}
%and transmits $X_{i}=\pi_{i}(\Sv_i)$, where $\pi_i$ is given by \eqref{control}.
%Then we have
%\[\sum_{j=1}^k x_{ji}=-C\sv_i=X_i\]
%where the last equality follows from (\ref{control}).
\item Decoders: At each time $i \in \{1, \ldots, n\}$ decoder $j$ forms an estimate $\Thetah_{ji}=-a_j^{-i}\Sh_{j(i+1)}$  for the desired message~$\Theta_j$, where $\Sh_{ji}$ is recursively formed as 
\begin{align}
\Sh_{j1} & =0  \nn \\
\Sh_{ji}&=a_j \Sh_{j(i-1)}+Y_{j(i-1)}, \quad i =2,3, \ldots,n. \label{BCdecoder}
\end{align}
%and $Y_{j(i-1)}$ denotes the channel output of receiver $j$ at time~$i-1$. %and is updated as in~\eqref{channel}.
\end{enumerate}
The following lemma characterizes the set of MSE exponent vectors that can be achieved by the sequence of $n$-codes so generated. 
\begin{lem}\label{MSEach}
%Let $A,B,C$ be same as (\ref{Aform}), (\ref{controlapplied}), (\ref{control}).
Let $\{\pi_i\}$ be a stabilizing control for~\eqref{BCsystem}. Then, the MSE exponent vector $(\log |a_1|, \ldots, \log |a_k| )$ is achievable under power constraint
\begin{equation*}
%\label{eq:minp}
P=\limsup_{n \to \infty} \frac{1}{n} \sum_{i=1}^n \E (\pi^2_i(\Sv_i)). %(\Sv_i )^2 .
\end{equation*}
\end{lem}
\IEEEproof See Appendix~\ref{applemMSEach}.

%Note that the set $\{\beta_j\}$ depends only on the open loop system dynamic and not on the controller. Therefore, any controller stabilizing the same system yields a different codes achieving the same set of MSE exponents.

%%%%%%%%%%%%%%%%%%%%%%%Section LQG %%%%%%%%%%%%%%%%%%%%%%
\subsection{LQG Code based on Optimal LQG Control}

According to Lemma~\ref{MSEach}, for a fixed open-loop matrix $A$, {\it any} stabilizing control yields a sequence of $n$-codes for the AWGN-BC with feedback that achieves the same MSE exponent vector $(\log|a_1|, \ldots, \log |a_k|)$ under a power constraint determined by the asymptotic control variance. The following theorem characterizes the performance of the LQG code, which corresponds to the minimum variance control that can be computed using the LQG control theory.
%A natural question to ask is what is the minimum asymptotic power required to stabilize the system $A$ or equivalently achieve the MSE exponent vector $(\log |a_1|, \ldots, \log |a_k| )$. 
\begin{thm}\label{thmLQGanalysis}
%Let $A =  \text{diag}(a_1,\ldots,a_k)$ where $a_j \in \Complex, \ j=1,\ldots,k$ are distinct points outside the unit circle, i.e., $|a_j|>1$. 
Let  $A$ and $B$ be given as in~\eqref{BCAform} and~\eqref{Bform}, and $K_z$ be the covariance matrix of the noise vector in the AWGN-BC~\eqref{channel}. Then, the rate vector $(\log |a_1|, \ldots, \log |a_k| )$ is achievable under power constraint~
\begin{equation}
\label{minpow}
P(A, K_z)=CK_sC'=\tr(GK_z)
\end{equation} 
where 
\begin{align}\label{optC}
C= (B'GB+1)^{-1}B'GA
\end{align}
and 
$G$ is the unique positive definite solutions to the discrete algebraic Riccati equation (DARE)
\begin{align}\label{riccati}
G=A'GA-A'GB(B'GB+1)^{-1}B'GA
\end{align}
such that $A-BC$ is stable, that is, every eigenvalue of $A- BC$ lies inside the unit circle, and $K_s$ is the unique solution to the discrete algebraic Lyapunov equation (DALE)
\begin{align}\label{Lyapunov}
K_s=(A-B C)K_s(A-BC)' + K_z.
\end{align} 
\end{thm}
%\IEEEproof The proof follows from Lemma~\ref{lemLQG} and Lemma~\ref{lempower} below combined with Lemma~\ref{lemMSE} and Lemma~\ref{MSEach}.

The proof of the theorem makes use of the following two lemmas.
\begin{lem}\label{lemLQG}
Given the unstable open-loop matrix~\eqref{BCAform}, the stationary and linear control
\begin{align}\label{BClincontrol}
X_i=-C\Sv_i
\end{align}
where $C$ and $G$ are given in~\eqref{optC} and~\eqref{riccati}, respectively, stabilizes the closed-loop system~\eqref{BCsystem} and minimizes the asymptotic average control power
\begin{align}\label{avconpower}
\limsup_{n \to \infty} \frac{1}{n}  \sum_{i=1}^n \E (\pi_i^2(\Sv_i)) .
\end{align}
%or equivalently the following discrete algebraic Lyapunov equation (DALE)
%\begin{align*}%\label{Lyapunov2}
%G=(A-BC)'G(A-BC)+C'C
%\end{align*}
The minimum stationary power is given by
\begin{align*}
\tr(GK_z)
\end{align*}
where $K_z$ is the covariance matrix of the noise vector in~\eqref{channel}.
\end{lem}
%\begin{remark}
%We refer to the linear code for the AWGN-BC with feedback, which is based on the optimal control given in Lemma~\ref{lemLQG}, as the {\it LQG code}. % given $A$. 
%\end{remark}
%\begin{remark}
%The achievable rate region by the LQG code is {\it implicitly} determined by the power constraint, that is, a rate vector is in the achievable region if the solution to the DARE~\eqref{riccati} satisfies the power constraint.   
%\end{remark}
\begin{IEEEproof}
Plugging \eqref{channel} into~\eqref{BCsystem}, we have
\begin{align}\label{BCclosedloop}
\Sv_i=A\Sv_{i-1}+B X_{i-1}+\Zv_{i-1}.
\end{align}
%The recursion~\eqref{BCclosedloop} represents a controlled system with an unstable open-loop matrix $A$, disturbance $\Zv_i \sim \N(0,K_z)$, and a scalar control signal $X_i$ applied to all component of the state vector (see~\eqref{Bform}).  
Consider the problem of finding the stabilizing control that minimizes the asymptotic average control power
\[
\limsup_{n \to \infty} \ \frac{1}{n}\sum_{i=1}^n \E(\pi^2_i(\Sv_i)).
\]
This problem is similar to the standard LQG problem~\cite{Chen} in the special case where the cost function does not depend on the state. For this problem, we can derive the Riccati equation~\eqref{riccati} and the stationary linear control~\eqref{optC}, similar to the solution to the LQG problem, to establish a sufficient condition for optimality in terms of the asymptotic power. %of a system with state $\Sv_i$, disturbance $\Zv_i$ and (scalar) control $X_i$. 
%Furthermore, we require that the control stabilizes the system in \eqref{BCsystem} . 

Unlike the standard LQG problem, though, here we require the control to stabilize the system (see Lemma~\ref{MSEach}). Next, we show that there exists a unique solution to~\eqref{riccati} such that the control $C$ in~\eqref{optC} is stabilizing, that is, $A-BC$ is stable. Since the eigenvalues of $A$ are all outside of the unit circle and the elements of $B$ are nonzero we know $(A,B)$ is detectable, that is, there
exists a $C \in \Complex^{1 \times k}$ such that $A-BC$ is stable. Then, by~\cite[Lemma 2.4]{YH} there exists a {\it unique} positive definite solution to~\eqref{riccati} for which $A-BC$ is stable. 
%Hence, we conclude that the stationary and linear control of the form~\eqref{optC} stabilizes the system~\eqref{BCsystem} and minimizes the asymptotic power. 
\end{IEEEproof}
%The optimal control to this LQG problem is characterized based on a Riccati recursion which converges to the Riccati equation~\eqref{riccati}. The stationary control~\ref{optC} then minimizes the stationary power 
%\[\lim_{n\to \infty} \frac{1}{n}\sum_{i=1}^n \E(\pi^2_i(\Sv_i)).\]
%Note that the solution to the LQG problem described above is $G=0, C=0$. However, we require and additional conststraint that the controller stabilizes the system in \eqref{BCsystem} (see Lemma~\ref{MSEach}. In fact, since the eigenvalues of $A$ are all outside of the unit circle, we know~\cite[Lemma 2.4]{YH} that there exists a unique positive definite solution to~\eqref{riccati} which stabilizes the closed loop system~\eqref{BCsystem}. 
%Note that $G=0$, which is a solution to \eqref{riccati}, is not acceptable since the corresponding controller $C=0$ does not stabilize the system in \eqref{BCsystem}, and we are only interested in the stabilizing controllers (see Theorem~\ref{MSEach}). In fact, since the eigenvalues of $A$ are all outside of the unit circle, we know~\cite[Lemma 2.4]{YH} that the stabilizing solution to~\eqref{riccati} is unique and positive-definite. 
%\begin{align}\label{riccati}
%G=A'GA-A'GB(B'GB+1)^{-1}B'GA
%\end{align}
%or equivalently the following discrete algebraic Lyapunov equation (DALE)
%\begin{align}\label{Lyapunov2}
%G=(A-BC)'G(A-BC)+C'C.
%\end{align}

The following lemma characterizes the asymptotic performance of the minimum variance control.
%The power of the stationary linear control, which stabilizes the system, alternatively can be obtained as follows.
\begin{lem}\label{lempower}
Let the linear control in~\eqref{BClincontrol} be stabilizing, that is, all eigenvalues of $A-BC$ lies inside the unit circle. Then, the asymptotic average control power~\eqref{avconpower} is given by
\begin{align}\label{powerLy}
CK_s C'
\end{align}
where $K_s$ is the unique solution to the following DALE
\begin{align*}%\label{Lyapunov}
K_s=(A-BC)K_s(A-BC)' + K_z.
\end{align*}
\end{lem}
\begin{remark}
From Lemma~\ref{lemLQG} and Lemma~\ref{lempower} it is clear that the performance of the described feedback code depends on the correlation among the noises at the receivers. 
\end{remark}
\begin{IEEEproof}
Plugging~\eqref{BClincontrol} into the closed-loop system dynamics~\eqref{BCclosedloop} we have 
\[\Sv_i=(A-BC)\Sv_{i-1}+\Zv_{i-1}.\]
Let $K_{s,i}$ denote the covariance matrix of the state $\Sv_i$, then we have the following discrete algebraic Lyapunov recursion
\[K_{s,i+1}=(A-BC)K_{s,i}(A-BC)'+K_z.\]
By assumption $(A-BC)$ is stable and $K_{s,0} \succ 0$, hence the above recursion converges to the unique positive-definite solution to the  following discrete algebraic Lyapunov equation (DALE)
\begin{align*}%\label{Lyapunov}
K_{s}=(A-BC)K_{s}(A-BC)'+K_z.
\end{align*}
Note that $X_i=-C\Sv_i$ and hence $\E(X^2_i)= CK_{s,i}C'$, which completes the proof. 
\end{IEEEproof}

%\IEEEproof The proof follows from Lemma~\ref{lemLQG} and Lemma~\ref{lempower} below combined with Lemma~\ref{lemMSE} and Lemma~\ref{MSEach}.

{\it Proof of Theorem~\ref{thmLQGanalysis}:}
According to Lemma~\ref{lemLQG} and Lemma~\ref{lempower}, for $A=\diag(a_1,\ldots,a_k)$ with $|a_j| > 1$, there exists a stabilizing control with asymptotic power equal to~\eqref{minpow}. Furthermore, according to Lemma~\ref{MSEach}, we can construct a sequence of $n$codes corresponding to this control, which achieves the MSE exponent vector $(\log |a_1|,\ldots, \log |a_k|)$ with  the same asymptotic power as the control. Finally, by Lemma~\ref{lemMSE}, we conclude that rate vector $(\log |a_1|,\ldots, \log |a_k|)$ is achievable under asymptotic power constraint~\eqref{minpow}. \hfill \IEEEQED

%%%%%%%%%%%%%%%%%%%%%%%%%%%%%%%%%%%%%%%%%%%%%%%%%%%
\begin{example}
Consider the special case of a two-receiver AWGN-BC, and let
\begin{align*}
A =  \left( \begin{array}{ccccc}
             a_1  &   0 \\
             0 &     a_2 \\
        \end{array} \right), \ K_z =  \left( \begin{array}{ccccc}
             1  &   \rho \\
             \rho &     1 \\
        \end{array} \right),
  \end{align*}
for $|a_1|>1$, $|a_2|>1$, and  $-1 < \rho < 1$.  By solving~\eqref{riccati} and plugging the solution into~\eqref{minpow} we obtain that $( \log |a_1|, \log |a_2| ) $ is an achievable rate pair under power constraint
\begin{align*}
 \frac{1}{|a_1-a_2|^2} \Big(|a_1a_2&-1|^2 (|a_1|^2+|a_2|^2-2) \nonumber \\
& - \rho (|a_1|^2-1)(|a_2|^2-1) (\tx{Re}(a_1a'_2)-2)\Big).
\end{align*}
%achieves the same rate pair under the same power constraint and hence it 
In the special case where the noises at the receivers are independent ($\rho = 0$), the code in~\cite{Elia2004} has the same performance as the LQG code.  
%From Theorem~\ref{MSEach} we know that any rate pair $(R_1,R_2)$ such that $R_1< \log(\b_1)$ and $R_1< \log(\b_1)$ are achievable. Moreover, from~\eqref{PN2} we can see (details needed) that to minimize the power we should pick $\theta_1+\theta_2=\pi$.
  \end{example}

%%%%%%%%%%%%%%%%%%%%%%%%%%%Section Symmetric %%%%%%%%%%%%%%%%%%%%%%%%
\section{Independent Noises: Power Gain}\label{SectionBCsymmetric}

In this section, we analyze the performance of the LQG code in the special case of independent noises, i.e., assuming that $K_z=I_{k \times k}$, which is the case for most practical scenarios. %Since $\rank(I)=k$, according to Theorem~\ref{prelog}, we have $\gamma(I) \leq 1$. Hence, we can not hope for any degrees of freedom gain and we focus on the power gain. 
We characterize the sum rate $R(k,P)$ achievable by the LQG code under power constraint~$P$ for the symmetric case where diagonal elements of $A$ in~\eqref{BCAform} are
\begin{equation}
\begin{array}{rl}
a_j = a e^{\frac{2\pi \sqrt{-1}}{k}(j-1)}, \quad j=1,\ldots, k.
\end{array}
\label{symmetricA}
\end{equation}
and $a > 1$ is real.

\begin{thm}\label{BCthmsym}
Given $A$ as in~\eqref{symmetricA}, the LQG code achieves symmetric sum rate $R(k,P)$, i.e., $R_j=R(k,P)/k$, $j=1,\ldots,k$, under power constraint $P$ where 
\[R(k,P) = \half \log (1+P\phi)\]
and $\phi(k,P)$ is the unique solution in the interval $[1,k]$ to
\[(1+P\phi)^{k-1}=\left(1+\frac{P}{k}\phi (k-\phi)\right)^k.\] 
\end{thm}
\begin{remark}
The quantity $\phi(k,P)$ is the power gain compared to  the no feedback sum capacity $1/2 \log (1+P)$. This power gain can be interpreted as the amount of cooperation among the receivers established through feedback, which allows the transmitter to align signals intended for different receivers coherently and use power more efficiently. 
\end{remark}

%We use the following lemma to prove Theorem~\ref{BCthmsym}. 
%{\em Proof of Theorem~\ref{BCthmsym}}: 
\begin{IEEEproof}
For $A$ defined as~\eqref{symmetricA}, we know by Theorem~\ref{thmLQGanalysis} that the sum rate $R=k\log a$ is achievable under power constraint $P=\tr(G)$ where $G$ is the unique solution to the DARE~\eqref{riccati}. The following lemma characterizes the solution to the DARE~\eqref{riccati} for the symmetric choice $A$. 

\begin{lem} \label{lemsym} \cite[Lemma 12]{Ehsan}
\label{symlem}
Suppose that the open-loop matrix $A$ is of the form~\eqref{symmetricA}.
Then the unique positive-definite solution $G$ to (\ref{riccati}) is circulant with real eigenvalues $\lambda_1>\lambda_2>\ldots>\lambda_k>0$ satisfying $$\lambda_i=\frac{1}{a^2} \lambda_{i-1}$$ for $i=2,\ldots, k$. The largest eigenvalue $\l_1$ satisfies
% is equal to $\sigma=\sigma_j=\sum_{k=1}^k \Kb_{jk}, \forall j$ and satisfies
\begin{align}
 1+k\lambda_1 &= a^{2k}  \label{symsum} \\
\Big(1+\lambda_1\big(k-\frac{\lambda_1}{G_{jj}}\big)\Big) &= a^{2(k-1)}. \label{symother}
\end{align}
%where $P=\Kt_{mm}$ is the main diagonal entries of $\Kt$.
\end{lem}
%\IEEEproof For completeness, the proof is provided in Appendix~\ref{appsymlem}.

From~\eqref{symsum} and~\eqref{symother} we have
\begin{align}\label{symsumother}
(1+k\l_1)^{k-1}=\left(1+\l_1 (k-\frac{\l_1}{G_{jj}})\right)^k.
\end{align} 
The solution to~\eqref{riccati} is unique and we conclude that~\eqref{symsumother} has a unique solution for $G_{jj}$ given $\l_1$ and vice-versa. Consider $\l_1(k,P)$ corresponding to the case where 
\[G_{jj}=\frac{P}{k}, \quad j=1,\ldots,k.\] 
Note that the solution $G$ to~\eqref{riccati} is circulant and has equal elements on the diagonal. From~\eqref{symsum} we know the LQG code corresponding to $\l_1(k,P)$ achieves sum rate
\[k\log a= \half \log(1+ k \l_1(k,P)).\]
under power constraint $\tr(G)=P$. The following change of variable 
%\[\l_1:=\frac{P}{k}\phi\]
\[\phi=\frac{k}{P} \l_1\]
completes the proof. 
\end{IEEEproof}
%\hfill \IEEEQED

\subsection{Comparison with the AWGN Multiple Access Channel}
The LQG approach can be also applied to the AWGN-MAC with feedback. It is known~\cite{EhsanControl} that the LQG code for AWGN-MAC has the same performance as the Kramer code~\cite{KramerFeedback}, which achieves the linear sum capacity~\cite{Ehsan}, the supremum sum rate achievable by linear codes. Let $R_{ \tx{\footnotesize MAC}}(k,P)$ denotes the symmetric sum rate achievable by the LQG code for the $k$-sender AWGN multiple access channel (MAC) with feedback where each sender has power constraint $P$. Then, we have~{\cite[Theorem~4]{EhsanControl}}:
\[R_{\tx{\footnotesize MAC}}(k,P) = \half \log (1+kP\phi)\]
where $\phi(k,P)$ is the unique solution to
\[(1+kP\phi)^{k-1}=\left(1+P\phi (k-\phi)\right)^k.\] 
Comparing with Theorem~\ref{BCthmsym}, it is not hard to see that
\[R_{\tx{\footnotesize BC}}(k,P)=R_{\tx{\footnotesize MAC}}(k,P/k).\]
This shows that under the same {\it sum power} constraint $P$, the sum rate achievable by the LQG code over MAC and BC are equal. This connection between the MAC and the BC is also considered in~\cite{Sriram}. 
\subsection{Ozarow-Leung (OL) code for $k=2$}
We compare the LQG code with the OL code for $k=2$ and $K_z=I$. The OL code can be represented as follows~\cite{YH--lecture}:
\[X_{i} = S_{1i}+ S_{2i}\]
where
\begin{align}~\label{Ozarowform}
\begin{bmatrix} 
             S_1  \\
             S_2 \\ 
\end{bmatrix}_{i+1}
& =
\begin{bmatrix} 
             a & 0  \\
             0 & -a \\ 
\end{bmatrix}
\left( 
\begin{bmatrix} 
             S_1  \\
             S_2 \\ 
\end{bmatrix}_{i}
- 
\begin{bmatrix} 
             \frac{\E[S_{1i}Y_{1i}]}{\E[Y_{1i}^2]}  & 0  \\
             0 & \frac{\E[S_{2i}Y_{2i}]}{\E[Y_{2i}^2]} \\ 
\end{bmatrix}
\begin{bmatrix} 
             Y_1  \\
             Y_2 \\ 
\end{bmatrix}_{i}
\right)
\end{align}
In the sequel, we present the LQG code in a similar form as~\eqref{Ozarowform} for comparison. Let the system in \eqref{BCsystem} be re-written as 
\[\Sv_{i}=A\Sv_{i-1}+ \diag(\Bt) \Yv_{i-1}.\] 
where \begin{align}
A = \begin{bmatrix} 
             a & 0  \\
             0 & -a \\ 
\end{bmatrix}, \quad  
\Bt = \begin{bmatrix} 
             -b  \\
             b \\ 
\end{bmatrix}
\label{bc1}
\end{align}
where $b \neq 0$ is any constant and $a > 1$. The choice of $b$ does not affect the performance of the optimal control as one can cancel out the effect of $b$ by properly scaling the control signal. Thus, without loss of generality, in~\eqref{BCsystem} we picked $b=1$. According to the channel~\eqref{channel} the closed-loop system is
\[\Sv_{i}=A\Sv_{i-1}+\Bt X_{i-1}+\diag(\Bt) \Zv_{i-1}.\] 
By substituting $B$ with $\Bt$, the optimal control can be characterized by Lemma~\ref{lemLQG} as follows. The solution to the DARE~\eqref{riccati} is 
\begin{align*}
G = - \frac{(a^4-1)}{4 a^4 b^2} \begin{bmatrix} 
             1+a^2 & 1-a^2 \\  1-a^2 & 1+a^2
\end{bmatrix}
%\label{macg}
\end{align*}
which yields the optimal control
\begin{align*}
C = [c_1 \ c_2]=- \frac{(a^4-1)}{2 a^3 b} \begin{bmatrix} 
             1 & 1\\
\end{bmatrix}.
%\label{macc}
\end{align*}
To obtain the power we need to substitute $K_z=I$ with 
\begin{align*}
Q = \diag(\Bt) \diag(\Bt)'
% \begin{bmatrix} 
%             -b & 0  \\
%             0 & b \\ 
%\end{bmatrix}  \begin{bmatrix} 
%             -b & 0  \\
%             0 & b \\ 
%\end{bmatrix}^\prime = 
=\begin{bmatrix} 
             b^2 & 0 \\
             0 & b^2 \\ 
\end{bmatrix} .% \label{bcz}
\end{align*}
and the asymptotic variance of the channel input $X$ is given by 
\begin{align}
P = \tr(GQ) = b^2 \tr(G) = \frac{1}{2a^2} (a^2-1) (a^2 + 1)^2.
\label{bcx}
\end{align}
Notice that  \eqref{bcx} does not depend on the parameter $b$. Thus, we can choose $b$ arbitrarily without affecting the overall performance of the code.  In particular, by choosing 
\[b = \frac{a^4-1}{2a^3}\] 
we have that $X_{i} = S_{1i}+ S_{2i}$ as in the OL code. However, the state in OL is updated as in~\eqref{Ozarowform} while in the LQG code,
\begin{align}\label{LQGform}
\begin{bmatrix} 
             S_1  \\
             S_2 \\ 
\end{bmatrix}_{i+1}
& =
\begin{bmatrix} 
             a & 0  \\
             0 & -a \\ 
\end{bmatrix}
\left( 
\begin{bmatrix} 
             S_1  \\
             S_2 \\ 
\end{bmatrix}_{i}
- 
\begin{bmatrix} 
             b/a & 0  \\
             0 & b/a \\ 
\end{bmatrix}
\begin{bmatrix} 
             Y_1  \\
             Y_2 \\ 
\end{bmatrix}_{i}
\right).
\end{align}
To compare with the OL code, we need to find the asymptotic covariance matrix of the state. By substituting $K_z$ with $Q$, the asymptotic covariance matrix $K_s$ is given by the DALE~\eqref{Lyapunov},
\begin{align*}
K_s = \frac{b^2}{2 a^2} 
\begin{bmatrix} 
             (a^2+1)^2+\frac{2}{a^2-1} & a^4+1 \\  a^4+1 & (a^2+1)^2+\frac{2}{a^2-1} \\
\end{bmatrix}
\end{align*}
and
\begin{align*}
\lim_{i \to \infty } \frac{\E[S_{2i}Y_{2i}]}{\E[Y_{2i}^2]}  &= \lim_{i \to \infty} \frac{\E[S_{1i}Y_{1i}]}{\E[Y_{1i}^2]}  \\
&= \frac{-c_1 (K_s)_{11} - c_2 (K_s)_{12} }{ \tr(GQ) +1 } \\
&= b \cdot \frac{a^3 (a^2-1)}{a^6 + a^4 + a^2 -1}\cdot %\label{mmse_bc}
\end{align*}
Notice that
\begin{align*}
b\cdot \frac{a^3 (a^2-1)}{a^6 + a^4 + a^2 -1} < \frac{b}{a}.
\end{align*}
Therefore, unlike in the point-to-point setting discussed in Section~\ref{SectionBCP2P}, here the OL code and the LQG code are not asymptotically equivalent. Although both codes achieve rate pair $(\log a, \log a)$, by Lemma~\ref{lemLQG}, the OL code requires more asymptotic power than the LQG code and hence it is suboptimal.

%Therefore, the asymptotic covariance of the state for the LQG code~\eqref{LQGform} is different from that of the OL code~\eqref{Ozarowform} in~\cite{Ozarow--Leung}. 

%Moreover, in MAC $\phi(k,P)$ captures the cooperation among the senders established through feedback, whereas, in BC it captures the amount of cooperation among receivers. In both problems, this amount of cooperation which can be established through feedback is the same under same sum power constraint.  
%%%%%%%%%%%%%%%%%%%%%%%%%%%%%%%%%%%%%%%%%%%%%%%%%%%%%%%
%%%%%%%%%%%%%%%%%%%%%%%%%%%%%%%%%%%%%%%%%%%%%%%%%%%%%

%%%%%%%%%%%%%%%%%%%%%%%%%%DOF
\section{Correlated Noises: Pre-Log Gain}\label{SectionBCDoF}
%In this section we characterize the minimum stationary power~\eqref{minpow} under which we can achieve the symmetric rate vector $(R,\ldots,R)$ using the code described above, and then 
%As shown in Theorem~\ref{thmLQGanalysis}, the performance of the LQG code depends on the noise covariance matrix $K_z$. 
In this section, we show that structured correlation among the noises at the receivers can increase the capacity significantly. We consider the high power regime and study the pre-log $\gamma(K_z)$ as a function of covariance matrix $K_z$, which represents the number of orthogonal point-to-point channels with the same sum capacity. 

\begin{thm}\label{prelog}
For all $K_z$ of rank $r$
\[\gamma(K_z) \leq k-r+1.\]
Conversely, for any $r=1,\ldots,k$, there exists $K_z$ such that $\rank(K_z)=r$ and 
\[\gamma(K_z)=k-r+1.\]
\end{thm}

\begin{IEEEproof}
First, we prove the upper bound by induction. By assumption $K_z$ contains $r$ linearly independent rows, let us assume, without loss of generality, the last $r$ rows. Assume that receivers $k-r+1, \cdots, k$ share their received signals and form a single receiver equipped with $r$ receive antennas and let $\Yv_{k-r+1}:=(Y_{k-r+1}, \ldots, Y_k)^T$ denote the vector of received signals by this multiple antenna receiver. The corresponding AWGN vector BC with feedback is specified by
\begin{align*}
Y_j & = X + Z_j, \qquad j=1,\cdots,k-r, \\
\Yv_{k-r+1} & = \1_{r \times 1} X + \Zv_{k-r+1}
\end{align*}
where $(Z_1,\cdots,Z_{k-r},\Zv_{k-r+1}) \sim \N \left(0, K_z \right)$, $\Zv_{k-r+1}  \sim \N(0, \tilde{K}_z)$, and by assumption $\Kt_z$ is full rank and invertible. 

Now suppose that the sender of this channel wishes to send message $M_j$ to receiver $j$, $j=1,\cdots, k-r+1$, under power constraint $P$.
%\allowdisplaybreaks
%\begin{align*}
%\sum_{i=1}^n \E[ X_{i}^2(m_1,\cdots,m_{k-r+1},&Y_{1}^{i-1},\cdots,Y_{k-r}^{i-1},\Yv_{k-r+1}^{i-1}) ]\\& \leq n P, \quad m_j \in [1:2^{nR_j}].
%\end{align*}
Since we made the optimistic assumption that a subset of receiver can cooperate, the sum capacity of this channel is an outer bound on the sum capacity of the original AWGN-BC. Note that for every $j=1,\cdots,k-r$, the rate
\[R_j <  \half \log (1+P)\]
is upper bounded by the capacity of the point-to-point AWGN channel.
%\begin{align*}
%nR_j %& \le H(M_j)  \\
%%& = I(M_j; Y_{j}^n) + n \epsilon_{n} \\
%%& = \sum_{i=1}^n I(M_j; Y_{ji} | Y_{j}^{i-1}) + n \epsilon_{n} \\
%%&\leq \sum_{i=1}^n I(M_j, X_j; Y_{ji} | Y_{j}^{i-1}) + n \epsilon_{n} \\ 
%%& \le \sum_{i=1}^n I(M_j, Y_{j}^{i-1} ; Y_{ji} ) + n \epsilon_{n} \\
%%& \le \sum_{i=1}^n I(X_{i}, M_1,\cdots,M_{N-k+1},Y_{1}^{i-1},\cdots,Y_{k}^{i-1},\bb{Y}_{N-k+1}^{i-1}; Y_{ji} ) + n \epsilon_{j,n} \\
%%& \substack{(a) \\ \leq} 
%& \leq \sum_{i=1}^n I(X_{i} ; Y_{ji} ) + n \epsilon_{n} \\
%& \le n \cdot \half \log (1+P) +  n \epsilon_{n},
%\end{align*}
%where $\e_n \to 0$ as $n \to \infty$. The first inequality follows since the channel is memoryless and the second inequality follows from the expected average power constraint on $X$. 
The rate $R_{k-r+1}$ for the $(k-r+1)$-th receiver with $r$ multiple antenna is upper bounded by the capacity of a single input multiple output (SIMO)~\cite{Tse} channel:      
%\begin{align*}
%nR_{k-r+1} & \le n \cdot \half \log(1+P |\tilde{K}^{-\frac{1}{2}} \bb{1}_{r\times 1} |^2 ) +  n \epsilon_{n}.
%\end{align*}
\begin{align*}
R_{k-r+1} & \le  \half \log(1+P |\tilde{K}^{-\frac{1}{2}} \1_{r\times 1} |^2 ) 
\end{align*}
where by assumption $\Kt$ is invertible. Thus, the sum capacity of this channel is upper bounded by $(k-r)/2 \log (1+P) + 1/2\log (1+P |\tilde{K}^{-\frac{1}{2}} \1_{r \times 1} |^2 )$, and therefore the pre-log $\gamma(K_z)$ can be at most $k-r+1$.

%For $r=1$, the upper bound 
%\[\kappa(K_z) \leq k\]
%is immediate by considering full cooperation among the receivers. Assume for $r$ we have
%\[\kappa(K_z) \leq k-r+1.\]
%For $r+1$ pick one of the receivers noise of which can not be written as a linear combination of the others. 

%We have the following cut-set upper bound on the sum rate
%\[R < I(X;Y_1,\ldots,Y_k)\]
%where $X \sim \N(0,P)$. Hence,
%\begin{align}
%R &< h(Y)-h(Y_1,\ldots,Y_k|X) \\
%& \half \log \frac{K_y}{K_z} 
%\end{align}

Next, we show $\gamma(K_z) = k$ is achievable by the LQG code for some $K_z$ of rank one, i.e., $r=1$. 
For $r=2,\ldots,k$ similar argument holds. Suppose that the open-loop matrix~$A$ is as in~\eqref{symmetricA}. By Theorem~\ref{thmLQGanalysis}, 
the symmetric rate vector $(\log a,\ldots,\log a)$ is achievable under the power constraint $\tr(GK_z)$, where $G$ is the circulant matrix in Lemma~\ref{lemsym}. Note that any circulant matrix can be written as $F\Lambda F'$, where $F$ is the $k$ point discrete Fourier transform matrix with
\begin{align*}
F_{jl}=\frac{1}{\sqrt{k}}e^{-2\pi \sqrt{-1} (j-1)(l-1)/k},
\end{align*}
for $j,k \in\{1,\cdots,k\}$, and $\Lambda=\mbox{diag}([\lambda_1, \ldots, \lambda_k])$ is the matrix with eigenvalues on its diagonal. Suppose that the noise covariance matrix is also a circulant matrix with diagonal entries equal to $1$. In particular, let 
\begin{align}\label{specialcov}
\Kt_z=F\Lamt F' , \quad  \Lamt=\mbox{diag}([0, \ldots 0, k]).
\end{align} 
Then, we have
\begin{align*}
\tr(G\Kt_z)&=\tr(F\Lambda \Lamt F')=k\lambda_k
\end{align*}
where by Lemma~\ref{lemsym},
\[ k\l_k=\frac{k \lambda_1}{a^{2(k-1)} } = \frac{a^{2k}-1}{a^{2(k-1)}}.\]
Therefore, for the symmetric choice of $A$ in~\eqref{symmetricA}, the LQG code achieves sum rate 
\begin{align}\label{DoFR}
R=k\log a
\end{align}
under power constraint 
\begin{align}\label{DoFP}
P=\tr(G\Kt_z)=\frac{a^{2k}-1}{a^{2(k-1)}}
\end{align}
and we have
\begin{align}
\gamma(\Kt_z)&= \lim_{P \to \infty} \frac{C(P,\Kt_z)}{\frac{1}{2} \log (1 + P )} \nn \\
&\geq \lim_{P \to \infty} \frac{R}{\frac{1}{2} \log (1 + P )} \nonumber\\
%&=\lim_{a \to \infty} \frac{k \log |a|}{ \half \log \left(k\frac{a^{2k}-1}{a^{2(k-1)}}\right) }  \\
&=\lim_{a \to \infty} \frac{k \log a}{ \half \log a^2}  \label{DoFplugP}\\
&=  k. \nn
\end{align}
where~\eqref{DoFplugP} follows by plugging sum rate $R$ and power $P$ from~\eqref{DoFR} and~\eqref{DoFP}. Moreover, we know $\gamma(\Kt_z) \leq k$ since $\rank(K_z)=1$. Hence, $\gamma(\Kt_z)=k$ for the covariance matrix $\Kt_z$ in~\eqref{specialcov}.

%For rank $r =2,\ldots, k$, consider $K_z$ such that it contains a $(k-r+1)\times(k-r+1)$ sub-matrix of rank one. 
To complete the proof, we show that  for every $r \in \{2,\ldots, k-1\}$ a pre-log equal to $k-r+1$ is achievable for some $K_z$ such that $\rank{K}_z=r$. Consider
$$
 K_z = \begin{bmatrix}
M_{ k-r+1, k-r+1 }   &  0_{ k-r+1, r-1 }  \\   
0_{ r-1, k-r+1 }    & I_{ r-1, r-1 } 
 \end{bmatrix}
$$
%$$ 
%K_z  =  [   M_{ k-r+1, k-r+1 }    0_{ k-r+1, r-1 }  ;    0_{ r-1, k-r+1 }    I_{ r-1, r-1 }   ]
%$$ 
where $0_{i,j}$ denotes the zero matrix of dimension $i \times j$, $I_{ i,i }$ is the identity matrix of dimension $i$, and $M_{ k-r+1, k-r+1 }$  is the $(k-r+1) \times (k-r+1)$ circulant matrix having first row equal to the last column of the discrete Fourier transform matrix of dimension $(k-r+1) \times (k-r+1)$.  Clearly, $\rank( K_z ) = \rank( M_{k-r+1,k-r+1} ) + \rank( I_{r-1,r-1} ) = r$. On the other hand, suppose that the transmitter communicates only to users $1,2, \cdots , k-r+1$, while the transmission rate for the remaining users is set to zero. %By setting the last $r-1$ components of $\diag(A)$, $B$, and $C$ equal to zero, 
We can use a similar argument as above and show that the LQG code for the corresponding $(k-r+1)$-receiver AWGN-BC with feedback and noise covariance matrix $M_{ k-r+1, k-r+1 }$ achieves a pre-log equal to $k-r+1$. 
\end{IEEEproof}
\begin{remark}
To achieve $\gamma=k$, we used the LQG code. However, the same pre-log can be achieved even with codes which are less power efficient since we are considering only the pre-log of the sum rate in the high power regime. For instance, for the special case of $k=2$, Gastpar and Wigger~\cite{Wigger} showed that the OL code, which is suboptimal, achieves pre-log two for anti-correlated noises. 
\end{remark}

%\begin{example}{\it (Independent Noises)}\\
%Consider the special case where the noises at the receivers are independent, i.e., $K_z=I_{k\times k}$. By combining Theorem~\ref{thmLQGanalysis} and Lemma~\ref{symlem}, it follows that the symmetric rate vector $(\log a,\ldots,\log a )$ is achievable under block power constraint 
%\begin{align*}
%\tr(G) &=\sum_{j=1}^k \l_j \nn\\
%&=\l_1 \sum_{j=1}^k a^{-2(j-1)} \nn\\
%&=\l_1 \frac{1-a^{-2k}}{1-a^{-2}}\nn \\
%&=\frac{a^{2k}-1}{k} \cdot \frac{1-a^{-2k}}{1-a^{-2}}.\nn
%\end{align*}
%Thus, in this case
%\begin{align*}
%\lim_{a \to \infty} \frac{k \log a}{ \tr(G)} = 1.
%\end{align*}
%\end{example}

\section{Conclusion}~\label{SectionBCcon}
Using tools from control theory we have presented a code for the $k$-user AWGN-BC with feedback, called the LQG code, which we have then used to investigate some properties of the capacity region of this channel. 
When the noises at the receivers are independent the pre-log of the sum capacity is at most one, so feedback can yield at most a power gain over the case without feedback. We have quantified the power gain achieved by the LQG code and shown that in the case where $k=2$, the LQG code recovers a previous result of Elia which strictly improves upon the OL code. In the case where the noises at the receivers are correlated, instead, the pre-log of the sum capacity can be strictly greater than one. We established that for all noise covariance matrixes of rank $r$ the pre-log is at most $k-r+1$ and, conversely, there exists a covariance matrix for which this upper bound is achieved by the LQG code. In particular, a pre-log equal to $k$ is achievable for some circulant noise covariance matrix of rank one. This generalizes previous results obtained by Gastpar and Wigger for the case $k=2$. 

The LQG approach exploited here could be in principle useful for other multi-user communication channels with feedback, when the subclass of linear codes can lead to optimal or close to optimal solutions. 

\appendices
%%%%%%%%%%%%%%%%%%%%%%%%%%%Appendix Proof of Lemma MSE

\section{Proof of Lemma~\ref{lemMSE}}\label{applemMSE}
The proof is closely related to the proof of an analogous statement for the communication problem over the multiple access channel with feedback~\cite{EhsanControl} . 
%Let MSE exponent vector $(E_1,\ldots,E_k)$ be achievable under power constraint $P$, that is, 
By assumption, there exists a sequence of $n$-codes for $\Theta_j \sim \U(0,1)$, such that 
\begin{align}\label{appMSEdef}
E_j = \lim_{n \to \infty} -\frac{1}{2n}\log \Dn_j, \quad j=1,\ldots, k 
\end{align}
and~\eqref{avP} holds. 
Given the sequence of $n$-codes and $R_j < E_j$,  $j=1\ldots,k$, we construct a sequence of $(2^{nR_1},\ldots, 2^{nR_k}, n)$ codes such that $\lim_{n \to \infty} \pen=0$. 

First, we map the discrete message $m_j \in \Mc_j=[1:2^{nR_j}]$ to a message point %$\Mc'=\{\theta_{i,n} \in (0,1)\}_{i=1}^{2^{\lceil nR \rceil}} $ 
$\theta_j(m_j) \in \mathbf{\Theta}_j $, %=\{\theta_{j,n}^{(i)} \in (0,1)\}_{i=1}^{2^{nR_j} } 
where $\mathbf{\Theta}_j$ is a set of $2^{nR_j}$ message points in the unit interval such that the distance between any two message points is greater than or equal to $2^{-nR_j}$. To send $m_j \in \Mc_j$, we use the given $n$-code and the corresponding message point $\theta_j (m_j)$. The decoder first forms the estimate of the message point $\thetah_j(y^n)$ according to the given $n$-code, and then chooses $\mh_j$ such that $\theta_j(\mh_j)$ is the closest message point to $\thetah_j(y^n)$. As the distance between any two message points is greater than or equal to $2^{-nR_j}$, the average probability of error is bounded as follows, 
\begin{align}\label{avebymax}
\pen \leq \max_{j} \max_{\theta\in \mathbf{\Theta}_j} \P\left\{ |\Theta_j- \Thetah_j| > \frac{1}{2} \cdot 2^{-nR_j} \big| \Theta_j=\theta\right\}.
\end{align}
To show $\lim_{n \to \infty} \pen=0$, consider
\begin{align}
p_{j,n}:=&\P\left\{|\Theta_j-\Thetah_j| >\half \cdot 2^{-nR_j}\right\} \\ 
&\leq 4\cdot2^{2nR_j}\cdot \Dn_j \label{Che} \\
& \leq 4\cdot2^{2nR_j}\cdot 2^{-2n(E_j-\e_n)} \label{Edefref} \\
&=4\cdot 2^{-2n(E_j-R_j-\e_n)} \label{RlessE}
\end{align}
where $\e_n \to 0$ as $n \to \infty$. The inequalities~\eqref{Che} and~\eqref{Edefref} follow from the Chebyshev inequality and~\eqref{appMSEdef}, respectively. From~\eqref{RlessE} and the assumption $R_j < E_j$, we have 
\begin{align}\label{intervaldec}
p_{j,n} \to 0 \ \tx{as} \ n \to 0 \quad j=1,\ldots,k.
\end{align}

%\begin{align}\label{intervaldec}
%\P \left\{ |\Theta-\Thetah_n| > \half \cdot 2^{-nR_0} \right\} \to 0 \ \tx{as} \ n\to \infty.
%\end{align}

%Recall that rate $R$ is defined based on a sequence of $(2^{nR},n)$ codes with discrete message set $\Mc=\{1,\ldots, 2^{\lceil nR \rceil}\}$. 

Next, by the similar argument as in~\cite[Lemma II.3]{Oferarxiv} we show that condition~\eqref{intervaldec} is sufficient to prove that there exists a set of message points in the unit interval such that the distance between any two message points is greater than or equal to $2^{-nR_j}$ and 
\begin{align}\label{maxerror}
\lim_{n \to \infty} \max_{\theta_j \in \mathbf{\Theta}_j} \P\left\{|\Theta_j- \Thetah_j| > \frac{1}{2} \cdot 2^{-nR_j} \Big| \Theta_j=\theta_j \right\}=0
\end{align}
%Define
%\begin{align}
%p_n=\P\left\{|\Theta-\Thetah| >\half \cdot 2^{-nR_0}\right\}.
%\end{align}
for $j=1,\ldots,k$. 

Define the event 
\begin{align*}
T_{j,n}=\Big\{\theta \in &(0,1): \\
& \P \Big\{|\Theta_j- \Thetah_j| > \frac{1}{2} \cdot 2^{-nR_j} \big|\Theta_j=\theta\Big\}> \sqrt{p_{j,n}}\Big\}.
\end{align*}
Then we have $p_{j,n} > \sqrt{p_{j,n}} \P(T_{j,n})$ and hence 
\[\P(T_{j,n}) < \sqrt{p_{j,n}}.\] 
To choose $\mathbf{\Theta_j}$ such that $\mathbf{\Theta}_j  \cap T_{j,n} =\emptyset$ and also the distance between any two message points is greater than or equal to $2^{-nR_j}$, it is sufficient that $|\mathbf{\Theta}_j| 2^{-nR_j} \leq 1-\sqrt{p_{j,n}}$ or considering~\eqref{intervaldec}, 
\[|\mathbf{\Theta}_j| \leq (1-\e_n) \cdot 2^{nR_j}\]
where $\e_n \to 0$ as $n \to \infty$. Moreover, by the definition of $T_{j,n}$ and the fact that $\mathbf{\Theta}_j \cap T_n= \emptyset$ we have
\begin{align}
\max_{\theta \in \mathbf{\Theta}_j} \P\left\{|\Theta- \Thetah| > \frac{1}{2} \cdot 2^{-nR_j} \Big| \Theta_j=\theta\right\} \leq \sqrt{p_{j,n}}
\end{align}
and considering~\eqref{intervaldec}, the condition~\eqref{maxerror} holds.  
%\[\lim_{n \to \infty} \max_{m \in \tilde{\Mc'}} \P(\theta \notin \Delta_n(Y^n)| M=m)=0.\]

Combining~\eqref{avebymax} and~\eqref{maxerror}, we have $\lim_{n \to \infty} \pen = 0$. Moreover, since the given $n$-code satisfies the power constraint~$P$, the constructed $(2^{nR_1},\ldots, 2^{nR_k}, n)$ code also satisfies the same power constraint. Hence, we conclude that the rate vector $(R_1,\ldots, R_k)$ is achievable under power constraint~$P$. 

%Moreover, comparing~\eqref{asymP} and~\eqref{avP}, it is not hard to see that the block power constraint~\eqref{avP} is satisfied for sufficiently large $n$ if $P=\Pb$, which completes the proof. 
%Combining~\eqref{apppe} and \eqref{appavP} we conclude that the rate vector $(R_1,\ldots, R_k)$ is achievable under block power constraint $P$.

%%%%%%%%%%%%%%%%%%%%%%%%%%%%%%%%%%%%%%%%%%%%%%%%%%%%%

%%%%%%%%%%%%%%%%%%%Appendix Proof of Lemma MSEach

\section{Proof of Lemma~\ref{MSEach}}\label{applemMSEach}
Let $\Shv_i=(\Sh_{1i}, \ldots, \Sh_{ki})^T$ where $\Sh_{ji}$ is given in~\eqref{BCdecoder}. Then, we have   
\begin{align*}
\Shv_1&=0 \nonumber \\ 
\Shv_{i}&=A\Shv_{i-1}+\Yv_{i-1} \quad i=2,3,\ldots
\end{align*}
where $A=\diag(a_1,\ldots,a_k)$ is the same as in~\eqref{BCAform}. Considering the recursion for $\Shv_i$, we can rewrite the system dynamics given in~(\ref{BCsystem}) as
\begin{align}
\Sv_{i+1}&=A^i \Thetav+\Shv_{i+1} \nonumber\\
   %A \Sv_{i}+ \Yv_{i}\\
    %&=A^i \Sv_0+A \Sv'_{i-1}+ \Yv_{i-1}\\
    &=A^i (\Thetav - \hat{\Thetav}_i) \label{BCMSEproof}
\end{align}
where $\hat{\Thetav}_i=(\Thetah_{1i},\ldots,\Thetah_{ki})^T$ and the last equality follows from the decoder rule by which $\Thetah_{ji}= -a^i \Sh_{j(i+1)}$.
%Therefore, $\Sv'_i=\Sv_i-A^i \Sv_0 = \Sv_i-  A^i \Mv$,
%and $\hat{\Mv}_n=-A^{-n}\Sv'_n= \Mv - A^{-n}\Sv_n$. 
From~\eqref{BCMSEproof}, the MSE for the message $\Theta_{j}$ at time~$n$ is given by 
\begin{align}
\Dn_j&=\E(\Theta_j-\Thetah_{jn})^2 = |a_j|^{-2n} (K_{n+1})_{jj}. \label{Ejbeta}
\end{align}
where $K_n:=\mbox{Cov}(\Sv_n)$ is the covariance matrix of $\Sv_n$. The achievability of MSE exponent $E_j=\log(|a_j|)$ follows from~\eqref{Ejbeta} and the assumption of stability $\limsup_{n \to \infty}(K_n)_{jj} < \infty$. The asymptotic power follows from the fact that $X_i=\pi_i(\Sv_i)$.

\bibliographystyle{IEEEtran}
\bibliography{bibliography}

\end{document}

%% file: defns_ehsan.tex
\newcommand{\tx}{\mbox}{}
\newcommand{\nn}{\nonumber}

\usepackage{xspace}
\usepackage{bbm}
\usepackage{mathrsfs}
\def\tr{\mathop{\rm tr}\nolimits}%
\def\diag{\mathop{\rm diag}\nolimits}%
\def\rank{\mathop{\rm rank}\nolimits}%
{}

\newcommand{\Cc}{\mathcal{C}}

\newcommand{\Mc}{\mathcal{M}}
\newcommand{\Nc}{\mathcal{N}}

\newcommand{\Cr}{\mathscr{C}}

\newcommand{\1}{{\bf 1}}

\newcommand{\Yv}{{\bf Y}}
\newcommand{\Zv}{{\bf Z}}

\newcommand{\Sv}{{\bf S}}

\newcommand{\Mv}{{\bf M}}

\newcommand{\Thetav}{\mathbf{\Theta}}
\newcommand{\pen}{{P_e^{(n)}}}
\newcommand{\Dn}{D^{(n)}}

\newcommand{\Mh}{{\hat{M}}}
\newcommand{\Sh}{{\hat{S}}}

\newcommand{\mh}{{\hat{m}}}

\newcommand{\Thetah}{\hat{\Theta}}
\newcommand{\thetah}{\hat{\theta}}
\newcommand{\Bt}{{\tilde{B}}}

\newcommand{\Kt}{{\tilde{K}}}

\def\e{\epsilon}

\def\l{\lambda}

\DeclareMathOperator\E{\mathsf{E}}
\let\P\relax
\DeclareMathOperator\P{\mathsf{P}}
\newcommand{\N}{\mathrm{N}}
\newcommand{\U}{\mathrm{Unif}}

\def\textiid{i.i.d.\@\xspace}
\newcommand\iid{\ifmmode\text{ i.i.d. } \else \textiid \fi}

\newcommand{\Complex}{\mathbb{C}}
\newcommand{\Real}{\mathbb{R}}
\newcommand{\half}{\frac{1}{2}}